\documentclass[10pt,a4paper]{article}

\usepackage{a4wide}
\usepackage{amsmath}
\usepackage{amsfonts}
\usepackage{amssymb}
\usepackage{amsthm}
\usepackage{amsopn}
\usepackage{graphicx}
\usepackage{wasysym}
\usepackage{marvosym}
\usepackage{titlesec}
\usepackage[dvipdfm,
            colorlinks=true,
            unicode,
            linkcolor=red,
            urlcolor=green,
            pagebackref=true,
            pdfpagemode=none,
            pdfstartview=FitH
            ]{hyperref}
\usepackage{color}

\DeclareMathOperator{\wds}{WDS}
\DeclareMathOperator{\diam}{diam}

\begin{document}

\newtheorem{definition}{Definition}[section]
\newtheorem{proposition}{Proposition}[section]
\newtheorem{conclusion}{Conclusion}[section]
\newtheorem{lemma}{Lemma}[section]
\newtheorem{theorem}{Theorem}[section]
\newtheorem{thmprime}{Theorem}[section]
\newtheorem{corollary}{Corollary}[section]
\renewcommand{\thecorollary}{}

\begin{titlepage}
\title{
Completeness of the WDS method in Checking Positivity of Integral Forms
\footnote{
Partially supported by a National Key Basic Research Project of China
(2004CB318000) and by National Natural Science Foundation of China
(10571095)
}
}
\date{}
\author{
Xiaorong Hou\footnote{The author to whom all correspondence should be sent.}, Junwei Shao\\
\textit{\small School of Automation Engineering,
University of Electronic Science and Technology of China, Sichuan, PRC}\\
\textit{\small E-mail: \href{mailto:houxr@uestc.edu.cn}{houxr@uestc.edu.cn},
\href{mailto:junweishao@gmail.com}{junweishao@gmail.com} }
}
\end{titlepage}
\maketitle

\noindent\textbf{Abstract:}
Examples show that integral forms can be efficiently proved positive semidefinite by the WDS method,
but it was unknown that
how many steps of substitutions are needed,
or furthermore, which integral forms is this method applicable for.
In this paper, we give upper bounds of step numbers of WDS
required in proving that an integral form is positive definite, positive semidefinite,
or not positive semidefinite,
thus deducing that the WDS method is complete.\\[2mm]
\textbf{Key words:}
integral form; weighted difference substitution; positivity\\[2mm]
\textbf{AMS subject classification(2000):} 26C99, 26D99

\section{Introduction}
Polynomials play key roles in many fields of the system theory \cite{bose:1},
fundamental problems in automatic control, filter theory and network realization
need to check some properties of polynomials, and positivity of
polynomials is an important one of such properties \cite{bose:2}.
When checking positivity of polynomials using traditional methods for
proving inequalities, complexities of algorithms are increasing rapidly as
variable number increases \cite{yang:1}.
Nowadays, Lu Yang \cite{yang:6} proposed a concise method to prove positivity of
homogeneous polynomials (i.e., forms), that is \textbf{difference substitution} (DS),
or its varied form \textbf{weighted difference substitution} (WDS).
This method demonstrates great efficiency.
\cite{yao:1} further showed that
if a form is indeed \textbf{positive definite} (PD) or not \textbf{positive semidefinite} (PSD),
then these properties can be checked by finite steps of WDS.
For integral forms, we estimate in this paper upper bounds of step numbers
required in checking these properties,
they only depend on the variable numbers, the degrees and the upper bounds of absolute values of
coefficients of the forms.
Therefore, we can also prove whether an integral form is PD through finite steps of WDS.

\section{Main Result}
We first introduce following definitions and notations according to \cite{yao:1}.

Considering $T_n \in \mathbb{R} ^{n\times n}$, where
\begin{equation}
T_n = \left(
\begin{array}{cccc}
1 & \frac{1}{2}& \ldots & \frac{1}{n} \\[2pt]
0 & \frac{1}{2} & \ldots & \frac{1}{n}\\[2pt]
\vdots & \ddots& \ddots& \vdots\\[2pt]
0 & \ldots & 0 & \frac{1}{n}
\end{array}
\right).
\end{equation}
Denote by $\Theta_n$ the set of all $n!$ permutations of $\{1,2,\ldots,n\}$,
for $(k_1 k_2 \ldots k_n) \in \Theta_n$,
let $P_{(k_1k_2...k_n)}=(a_{ij})_{n \times n}$ be the permutation corresponding
$(k_1 k_2 \ldots k_n)$, that is
$$
a_{ij}=\left\{
\begin{array}{ll}
1, & j=k_i\\
0, & j \neq k_i
\end{array}
\right..
$$
Let
$$A_{(k_1k_2...k_n)}=P_{(k_1k_2...k_n)}T_n,$$
and call it the WDS matrix
determined by permutation $(k_1k_2...k_n)$,
denote by $\Gamma_n$ the set of all $n!$ such matrices.
The variable substitution $\mathbf{x} = A_{(k_1 k_2 \ldots k_n)} \mathbf{y}$ corresponding
$(k_1 k_2 \ldots k_n)$ is called a WDS,
where $\mathbf{x} = (x_1,x_2,\ldots,x_n)^T, \mathbf{y} = (y_1,y_2,\ldots,y_n)^T$,
the following set of substitutions
\begin{equation}
\{ \mathbf{x} = A_1 A_2 \cdots A_m \mathbf{y}: A_i \in \Gamma_n \} \nonumber
\end{equation}
is called the $m$-th WDS set.

Suppose $A=(a_{ij}) \in \mathbb{R}^{n \times n}$,
we call it a normal matrix, and the corresponding substitution a normal substitution
if $\sum\limits_{i=1}^n a_{ij}=1, j=1,2, \ldots n$.
Thus WDS matrices are normal matrices and WDS substitutions are normal substitutions.

\begin{lemma} \label{normalmatrix}
Let $A = (a_{ij})_{n \times n} = B_1 B_2 \cdots B_k$,
where $B_i(i=1,2,\ldots, k)$ are all normal matrices.
Then $A$ is a normal matrix,
and for the substitution $\mathbf{x}=A \mathbf{y}$,
we have $\sum\limits_{i=1}^{n} x_i = \sum\limits_{i=1}^{n} y_i$.
\end{lemma}

\begin{proof}
Suppose $B_1 = (b_{1ij}),B_2 = (b_{2ij})$ are normal matrices,
let $C=B_1 B_2$, and denote by $C=(c_{ij})$,
then
$$
\sum\limits_{i=1}^n c_{ij} = \sum\limits_{i=1}^n \sum\limits_{k=1}^n b_{1ik} b_{2kj}
= \sum\limits_{k=1}^n \left( \sum\limits_{i=1}^n b_{1ik} \right) b_{2kj}
= \sum\limits_{k=1}^n b_{2kj}
= 1.
$$
Thus $C$ is normal, and further we can prove $A$ is normal by introduction.
Moreover, we have
$$
\sum\limits_{i=1}^{n} x_i
= \sum\limits_{i=1}^{n} \sum\limits_{j=1}^n a_{ij} y_j
= \sum\limits_{j=1}^n \left( \sum\limits_{i=1}^{n} a_{ij} \right) y_j
= \sum\limits_{j=1}^n y_j.
$$

\end{proof}

Let $f(\mathbf{x}) \in \mathbb{R}[x_1,x_2,\ldots,x_n]$ be a form, we call
\begin{equation}
\wds(f)=\bigcup\limits_{\theta \in \Theta_n} \{ f(A_{\theta} \mathbf{x}) \}
\end{equation}
the WDS set of $f$,
\begin{equation}
\wds^{(m)}(f)=\bigcup\limits_{\theta_m \in \Theta_n} \cdots \bigcup\limits_{\theta_1 \in \Theta_n}
\{ f(A_{\theta_m} \cdots A_{\theta_1} \mathbf{x}) \}
\end{equation}
the $m$-th WDS set of $f$ for positive integer $m$, and set $\wds^{(0)}(f)=\{ f \}$.

Denote by $\mathbb{N}$ the set of nonnegative integers,
let $\alpha=(\alpha_1,\alpha_2,\ldots,\alpha_n) \in \mathbb{N}^{n}$,
$|\alpha|=\alpha_1+\alpha_2+ \cdots +\alpha_n$.
For a form of degree $d$
$$
f(x_1,x_2,\ldots,x_n)=\sum\limits_{|\alpha|=d} c_\alpha x_1^{\alpha_1}x_2^{\alpha_2}\cdots x_n^{\alpha_n},
$$
if all coefficients $c_\alpha$ are nonzero, we say $f$ has complete monomials.

Let
$$
\mathbb{R}_+^n = \{ (x_1,x_2,\ldots,x_n)^T: x_i \geq 0, i=1,\ldots,n \},
$$
denote the $(n-1)$-dimensional simplex in $\mathbb{R}^n$ by
$$
\Delta_{n} = \left\{
(x_1,x_2,\ldots,x_n)^T: \sum\limits_{i=1}^n x_i = 1, (x_1,x_2,\ldots,x_n)^T \in \mathbb{R}_+^n
\right\}.
$$
and let
$$
\widetilde{\Delta}_{n} = \left\{
(x_1,x_2,\ldots,x_n)^T: \sum\limits_{i=1}^n x_i \leq 1, (x_1,x_2,\ldots,x_n)^T \in \mathbb{R}_+^n
\right\},
$$

\begin{definition}
Suppose $D \subset \mathbb{R}^n$, $f(\mathbf{x}) \in \mathbb{R}[x_1,x_2,\ldots,x_n]$,
$f(\mathbf{x})$ is PD in $D$ if $f(\mathbf{x})>0$ for any $\mathbf{x} \in D \setminus \{ \mathbf{0} \}$,
and it is PSD in $D$ if $f(\mathbf{x}) \geq 0$ for any $\mathbf{x} \in D \setminus \{ \mathbf{0} \}$.
\end{definition}

Obviously, we have,
\begin{lemma}  \label{lemma:equivalence}
A form $f \in \mathbb{R}[x_1,x_2,\ldots,x_n]$ has the same
positivity in $\mathbb{R}_{+}^{n}, \Delta_{n}$ and
$\widetilde{\Delta}_{n}$.
\end{lemma}

Denote by $\mathbb{Z}$ the set of integers.
We deduce the following result for integral forms.
\begin{theorem} \label{thm:poster}
Suppose $f \in \mathbb{Z}[x_1, x_2, \ldots, x_n]$ is a form of degree $d$,
and the absolute values of its coefficients do not exceed $M$,
then we have
\begin{enumerate}
\item
$f$ is PD in $\Delta_n$,
if and only if there exists $m \leq C_p(M,n,d)$,
such that each form in $\wds^{(m)}(f)$ has complete monomials, and its
coefficients are all positive;
\item
$f$ is not PSD in $\Delta_n$ (i.e., the minimum of $f$ in $\Delta_n$ is negative),
if and only if there exists $m \leq C_{nps}(M,n,d)$,
such that a form in $\wds^{(m)}(f)$ has complete monomials, and its
coefficients are all negative.
\end{enumerate}
where
\begin{eqnarray}
C_p(M,n,d)= \left[ \dfrac{ \ln \left( 2^{d^n} M^{d^n+1} n^{d^{n+1}+d} d^{(n+1)d+n d^n} (d+1)^{(n-1)(n+2)} \right)}
{\ln n - \ln (n-1)} \right] + 2 \\
C_{nps}(M,n,d)= \left[ \dfrac{ \ln \left( 2^{d^n+1} M^{d^n+1} n^{d^{n+1}+d} d^{(n+1)d+n d^n} (d+1)^{(n-1)(n+2)} \right)}
{\ln n - \ln (n-1)} \right] + 2
\end{eqnarray}
\end{theorem}

Thus, we can completely determine positivity of
$f$ through checking positivity of coefficients of forms in
$\wds^{(C_{nps}(M,n,d))}(f)$:
\begin{enumerate}
\item
If each form in $\wds^{(C_{nps}(M,n,d))}(f)$ has complete monomials,
and its coefficients are all positive, then $f$ is PD in $\Delta_n$;

\item
If each form in $\wds^{(C_{nps}(M,n,d))}(f)$ has a nonnegative coefficient,
then $f$ is PSD in $\Delta_n$;

\item
If there exists a form in $\wds^{(C_{nps}(M,n,d))}(f)$ has complete monomials,
and its coefficients are all negative, then $f$ is not PSD in $\Delta_n$.
\end{enumerate}

\section{Estimate for lower bounds of positive definite integral forms in the simplex}
\cite{basu:1} gives estimate for lower bounds of
positive definite integral polynomials in simplex,
\cite{jeronimoa:1} improves the estimate.

\begin{lemma}[\cite{jeronimoa:1}] \label{lemma:lowerbound}
Suppose $f \in \mathbb{Z}[x_1,x_2,\ldots,x_n]$ is positive definite in $\widetilde{\Delta}_{n}$.
If the degree of $f$ is $d$,
and absolute values of its coefficients do not exceed $M$,
then
\begin{equation}
\min\limits_{\widetilde{\Delta}_{n}} f \geq (2M)^{-d^{n+1}} d^{-(n+1)d^{n+1}}.
\end{equation}
\end{lemma}

Indeed, the deduction in \cite{jeronimoa:1} has proved the following more general result.

\begin{lemma}[\cite{jeronimoa:1}] \label{lemma:lowerboundg}
Suppose the minimum of $f \in \mathbb{Z}[x_1,x_2,\ldots,x_n]$ in $\widetilde{\Delta}_{n}$ is not zero.
If the degree of $f$ is $d$,
and absolute values of its coefficients do not exceed $M$,
then
\begin{equation}
\left| \min\limits_{\widetilde{\Delta}_{n}} f \right| \geq (2M)^{-d^{n+1}} d^{-(n+1)d^{n+1}}.
\end{equation}
\end{lemma}

We have the following result for integral forms in $\Delta_{n}$.
\begin{lemma} \label{lemma:lowerboundf}
Suppose the minimum of $f \in \mathbb{Z}[x_1,x_2,\ldots,x_n]$ in $\Delta_{n}$ is not zero.
If the degree of $f$ is $d$,
and absolute values of its coefficients do not exceed $M$,
then
\begin{equation}
\left| \min\limits_{\Delta_n} f \right| \geq C_1(M,n,d).
\end{equation}
where $C_1(M,n,d) = (2M)^{-d^{n}} n^{-d^{n+1}-d} d^{-n d^{n}}$.
\end{lemma}

\begin{proof}
Let $(x_{1,0},\ldots,x_{n,0})$ be a minimal point of $f$ in $\Delta_n$,
then $x_{j,0} \geq \dfrac{1}{n}, 1 \leq j \leq n$, we can suppose $x_{n,0} \geq \dfrac{1}{n}$
without loss of generality.
Thus
\begin{equation}
\begin{split}
\left| \min\limits_{\Delta_n} f \right| = & \left| f(x_{1,0},\ldots,x_{n,0}) \right| \\
= & (n x_{n,0})^d
\left| f \left( \dfrac{x_{1,0}}{n x_{n,0}}, \ldots, \dfrac{x_{n-1,0}}{n x_{n,0}}, \dfrac{1}{n} \right) \right| \\
\geq & \left| f \left( \dfrac{x_{1,0}}{n x_{n,0}}, \ldots, \dfrac{x_{n-1,0}}{n x_{n,0}}, \dfrac{1}{n} \right) \right|.
\end{split} \nonumber
\end{equation}

Let
\begin{equation}
\begin{split}
g(x_1,\ldots,x_{n-1}) = & n^d f(x_1,\ldots,x_{n-1},\frac{1}{n})\\
= & f(n x_1, \ldots, n x_{n-1}, 1),
\end{split} \nonumber
\end{equation}
then its minimum is not zero in $\widetilde{\Delta}_{n-1}$.
Degree of $g \in \mathbb{Z}[x_1,\ldots,x_{n-1}]$ is $d$,
and absolute values of its coefficients do not exceed $n^d M$,
so from Lemma \ref{lemma:lowerboundg}, we have
$$
\left| \min\limits_{\widetilde{\Delta}_{n-1}} g \right| \geq (2M)^{-d^{n}} n^{-d^{n+1}} d^{-n d^{n}}.
$$
Since
$$
\left( \dfrac{x_{1,0}}{n x_{n,0}}, \ldots, \dfrac{x_{n-1,0}}{n x_{n,0}} \right) \in \widetilde{\Delta}_{n-1},
$$
we have
\begin{equation}
\begin{split}
& \left| f \left( \dfrac{x_{1,0}}{n x_{n,0}}, \ldots, \dfrac{x_{n-1,0}}{n x_{n,0}}, \dfrac{1}{n} \right) \right| \\
\geq & n^{-d} (2M)^{-d^{n}} n^{-d^{n+1}} d^{-n d^{n}} \\
= & (2M)^{-d^{n}} n^{-d^{n+1}-d} d^{-n d^{n}}.
\end{split} \nonumber
\end{equation}
Therefore
$$
\left| \min\limits_{\Delta_n} f \right| \geq (2M)^{-d^{n}} n^{-d^{n+1}-d} d^{-n d^{n}}.
$$
\end{proof}

\section{WDS and barycentric subdivision}
In the $\Delta_{n}$ simplex coordinate system, considering a WDS
\begin{equation} \label{eqn:wds}
\mathbf{x} = T_n \mathbf{y},
\end{equation}
we can see that $\mathbf{a}_1=(1,0,...,0)^T$ is transformed to $(1,0,\ldots,0)^T$,
$\mathbf{a}_2=(\frac{1}{2},\frac{1}{2},\ldots,0)^T$ is transformed to $(0,1,\ldots,0)^T$,
$\ldots$,
and $\mathbf{a}_n=(\frac{1}{n},\frac{1}{n},\ldots,\frac{1}{n})^T$ is transformed to $(0,0,\ldots,1)^T$.
Moreover, $\mathbf{a}_k (k=1,2,\cdots,n)$ is the barycenter of the $(k-1)$-dimensional proper face containing
$\mathbf{a}_1 \mathbf{a}_2 \cdots \mathbf{a}_k$ in $\Delta_{n}$.
Since \eqref{eqn:wds} is a normal substitution,
from Lemma \ref{normalmatrix} we know, after transform \eqref{eqn:wds},
the corresponding point for any $(x_1,x_2,\cdots,x_n)^T \in \Delta_{n}$ satisfies
$\sum\limits_{i=1}^n y_i = 1$,
that is, coordinates after transforms are also normal.
So, $\mathbf{a}_1 \mathbf{a}_2 \cdots \mathbf{a}_n$ is a subsimplex of $\Delta_{n}$
after the first barycentric subdivision,
it corresponds a WDS matrix $T_n=(\mathbf{a}_1 \ \mathbf{a}_2 \ \cdots \ \mathbf{a}_n)$.

Similarly, other $n!-1$ WDS matrices respectively correspond other $n!-1$ subsimplexes of $\Delta_{n}$
after the first barycentric subdivision.
Thus, from geometrical views, a WDS corresponds a barycentric subdivision of $\Delta_{n}$.

From Lemma \ref{normalmatrix} and the definition of WDS, we know that
sequential WDS correspond sequential barycentric subdivisions of $\Delta_n$.

Denote by $\diam \sigma$ the dimension of simplex $\sigma$, i.e.,
maximal distance between vertexes of $\sigma$.
Comparing with the dimension of original simplex,
dimensions of subsimplexes in barycentric subdivision decrease. That is

\begin{lemma}[\cite{spanier:1}] \label{lemma:diam}
Let $\sigma$ be an $n$-dimensional simplex, $\sigma'$ is a subsimplex in
the barycentric subdivision of $\sigma$,
then
\begin{equation}
\diam \sigma' \leq \dfrac{n}{n+1} \diam \sigma.
\end{equation}
\end{lemma}

\section{Proof of the main result}

\begin{proof}[Proof of the Theorem \ref{thm:poster}]
We will prove two propositions respectively.

(\uppercase\expandafter{\romannumeral 1})
Sufficiency is obvious.
Now we suppose $f$ is positive definite in $\Delta_n$.

Let
$$f(x_1,\ldots,x_n)=\sum\limits_{i_1+\ldots+i_n=d}c_{i_1 \cdots i_n}x_1^{i_1} \cdots x_n^{i_n},$$
choose an arbitrary $m$-th WDS for $f$
\begin{equation} \label{eqn:subs}
\left\{\begin{array}{l}
x_1  =  \alpha_{1}y_1+(\alpha_{1}+\beta_{12})y_2+\ldots+(\alpha_{1}+\beta_{1n})y_n\\
x_2  =  \alpha_{2}y_1+(\alpha_{2}+\beta_{22})y_2+\ldots+(\alpha_{2}+\beta_{2n})y_n\\
\dotfill\\
x_n  =  \alpha_{n}y_1+(\alpha_{n}+\beta_{n2})y_2+\ldots+(\alpha_{n}+\beta_{nn})y_n
\end{array}\right.,
\end{equation}
where $\alpha_i, \alpha_i+\beta_{ij} \geq 0, i=1,\ldots,n, j=2,\ldots,n$.
Since a WDS is a normal substitution,
\begin{equation} \label{eqn:alpha}
\sum\limits_{i=1}^n\alpha_i =1, \quad \sum\limits_{i=1}^n\beta_{ij} = 0, j=2,\ldots,n.
\end{equation}
From Lemma \ref{lemma:diam}, we have
$$
\sqrt{\sum\limits_{i=1}^n\beta_{ij}^2} \leq \left( \dfrac{n-1}{n} \right)^m,\quad j=2,\ldots,n,
$$
Further more,
\begin{equation} \label{eqn:betagamma}
\left| \beta_{ij} \right| \leq \left( \dfrac{n-1}{n} \right)^m,\quad
i=1,\ldots,n, j=2,\ldots,n.
\end{equation}

Let $u=y_1+\ldots+y_n$, then \eqref{eqn:subs} can be written as
\begin{equation}
\left\{\begin{array}{l}
x_1  =  \alpha_{1} u + \beta_{12} y_2 +\ldots + \beta_{1n} y_n\\
x_2  =  \alpha_{2} u + \beta_{22} y_2 +\ldots + \beta_{2n} y_n\\
\dotfill\\
x_n  =  \alpha_{n} u + \beta_{n2} y_2 +\ldots + \beta_{nn} y_n
\end{array}\right.. \nonumber
\end{equation}

So
\begin{equation}
\begin{split}
& f(x_1,\ldots,x_n)\\
=\ & f(\alpha_{1} u + \beta_{12} y_2 +\ldots + \beta_{1n} y_n,
\ldots,
\alpha_{n} u + \beta_{n2} y_2 +\ldots + \beta_{nn} y_n)\\
=\ &
\sum\limits_{i_1+\ldots+i_n=d}
c_{i_1 \cdots i_n}
(\alpha_{1} u + \beta_{12} y_2 +\ldots + \beta_{1n} y_n)^{i_1}
\cdots
(\alpha_{n} u + \beta_{n2} y_2 +\ldots + \beta_{nn} y_n)^{i_n} \\
=\ &
\sum\limits_{i_1+\ldots+i_n=d}
c_{i_1 \cdots i_n}
\alpha_{1}^{i_1} \cdots \alpha_{n}^{i_n} u^d +
\sum\limits_{i_1+\ldots+i_n=d}
\phi_{i_1 \cdots i_n}(\alpha_{1},\ldots,\beta_{nn})
y_1^{i_1} \cdots y_n^{i_n} \\
=\ &
f(\alpha_{1},\ldots,\alpha_{n}) u^d +
\sum\limits_{i_1+\ldots+i_n=d}
\phi_{i_1 \cdots i_n}(\alpha_{1},\ldots,\beta_{nn})
y_1^{i_1} \cdots y_n^{i_n} \\
=\ &
\sum\limits_{i_1+\ldots+i_n=d}
\left(
\frac{d!}{i_1! \cdots i_n!}
f(\alpha_{1},\ldots,\alpha_{n})+
\phi_{i_1 \cdots i_n}(\alpha_{1},\ldots,\beta_{nn})
\right)
y_1^{i_1} \cdots y_n^{i_n} \\
=\ &
\sum\limits_{i_1+\ldots+i_n=d}
\widetilde{c}_{i_1 \cdots i_n}
y_1^{i_1} \cdots y_n^{i_n},
\end{split}
\end{equation}
where
\begin{equation} \label{eqn:cijk}
\begin{split}
\widetilde{c}_{i_1 \cdots i_n}=
\frac{d!}{i_1! \cdots i_n!}
f(\alpha_{1},\ldots,\alpha_{n})+
\phi_{i_1 \cdots i_n}(\alpha_{1},\ldots,\beta_{nn}).
\end{split}
\end{equation}
and
\begin{equation}
\begin{split}
& \sum\limits_{i_1+\ldots+i_n=d}
\phi_{i_1 \cdots i_n}(\alpha_{1},\ldots,\beta_{nn})
y_1^{i_1} \cdots y_n^{i_n} \\
= &
\sum\limits_{i_1+\ldots+i_n=d \atop j_{11}+\ldots+j_{n1} \neq d}
c_{i_1 \cdots i_n}
\left(
\sum\limits_{j_{11}+\ldots+j_{1n}=i_1}
\frac{i_1!}{j_{11}! \cdots j_{1n}!}
\alpha_1^{j_{11}}\beta_{12}^{j_{12}} \cdots \beta_{1n}^{j_{1n}}
u^{j_{11}} y_2^{j_{12}} \cdots y_n^{j_{1n}}
\right) \cdots \\
& \left(
\sum\limits_{j_{n1}+\ldots+j_{nn}=i_n}
\frac{i_n!}{j_{n1}! \cdots j_{nn}!}
\alpha_n^{j_{n1}}\beta_{n2}^{j_{n2}} \cdots \beta_{nn}^{j_{nn}}
u^{j_{n1}} y_2^{j_{n2}} \cdots y_n^{j_{nn}}
\right) \\
= &
\sum\limits_{i_1+\ldots+i_n=d \atop j_{11}+\ldots+j_{n1} \neq d}
c_{i_1 \cdots i_n}
\sum\limits_{\tiny \begin{array}{c} j_{11}+\ldots+j_{1n}=i_1\\ \dotfill \\j_{n1}+\ldots+j_{nn}=i_n \end{array}}
\frac{ \prod_{k=1}^n i_k! }{ \prod_{p,q=1}^n j_{pq}!}
\prod_{k=1}^n \alpha_k^{j_{k1}}
\prod_{p=1}^n \prod_{q=2}^n \beta_{pq}^{j_{pq}}
\cdot \\
&
u^{j_{11}+\ldots+j_{n1}} y_2^{j_{12}+\ldots+j_{n2}} \cdots y_n^{j_{1n}+\ldots+j_{nn}} \\
= &
\sum\limits_{i_1+\ldots+i_n=d \atop j_{11}+\ldots+j_{n1} \neq d}
c_{i_1 \cdots i_n}
\sum\limits_{\tiny \begin{array}{c} j_{11}+\ldots+j_{1n}=i_1\\ \dotfill \\j_{n1}+\ldots+j_{nn}=i_n \end{array}}
\frac{ \prod_{k=1}^n i_k! }{ \prod_{p,q=1}^n j_{pq}!}
\prod_{k=1}^n \alpha_k^{j_{k1}}
\prod_{p=1}^n \prod_{q=2}^n \beta_{pq}^{j_{pq}}
\cdot \\
& \sum\limits_{s_1+\ldots+s_n=j_{11}+\ldots+j_{n1}}
\frac{(j_{11}+\ldots+j_{n1})!}{s_1! \cdots s_n!}
y_1^{s_1} y_2^{j_{12}+\ldots+j_{n2}+s_2} \cdots y_n^{j_{1n}+\ldots+j_{nn}+s_n},
\end{split} \nonumber
\end{equation}
from \eqref{eqn:alpha} and \eqref{eqn:betagamma}, we have
\begin{equation}
\begin{split}
& \left|
c_{i_1 \cdots i_n}
\frac{ \prod_{k=1}^n i_k! }{ \prod_{p,q=1}^n j_{pq}!}
\prod_{k=1}^n \alpha_k^{j_{k1}}
\prod_{p=1}^n \prod_{q=2}^n \beta_{pq}^{j_{pq}}
\frac{(j_{11}+\ldots+j_{n1})!}{s_1! \cdots s_n!}
\right| \\
\leq & M (d!)^{n+1} \left(\frac{n-1}{n}\right)^{m \sum_{p=1}^n \sum_{q=2}^n j_{pq}}\\
= & M (d!)^{n+1} \left(\frac{n-1}{n}\right)^{m(d-\sum_{p=1}^n j_{p1})}\\
\leq & M (d!)^{n+1} \left(\frac{n-1}{n}\right)^m\\
\leq & M d^{(n+1)d} \left(\frac{n-1}{n}\right)^m.
\end{split} \nonumber
\end{equation}
There are
$$
{d+n-1 \choose n-1} \leq (d+1)^{n-1}
$$
nonnegative integer tuples $(i_1,\ldots,i_n)$ satisfying $i_1+\ldots+i_n=d$,
so $\phi_{i_1 \cdots i_n}(\alpha_{1},\ldots,\beta_{nn})$ is summed up by terms whose number does not exceed
$$
(d+1)^{(n-1)(n+2)},
$$
and the absolute value of each term does not exceed
$$
M d^{(n+1)d} \left(\frac{n-1}{n}\right)^m.
$$
Therefore
\begin{equation} \label{eqn:boundphi}
\left|
\phi_{i_1 \cdots i_n}(\alpha_{1},\ldots,\beta_{nn}))
\right|
\leq
M d^{(n+1)d} (d+1)^{(n-1)(n+2)} \left(\frac{n-1}{n}\right)^m.
\end{equation}

From Lemma \ref{lemma:lowerboundf}, we have
\begin{equation} \label{eqn:boundf}
\frac{d!}{i_1! \cdots i_n!} f(\alpha_{1},\ldots,\alpha_{n}) \geq C_1(M,n,d).
\end{equation}

From \eqref{eqn:cijk}, \eqref{eqn:boundphi}, \eqref{eqn:boundf}, we know that
in order that $\widetilde{c}_{ijk} > 0$, it suffices
$$
C_1(M,n,d) > M d^{(n+1)d} (d+1)^{(n-1)(n+2)} \left(\frac{n-1}{n}\right)^m.
$$
That is
\begin{equation} \label{eqn:c2}
m > \dfrac{ \ln \left( 2^{d^n} M^{d^n+1} n^{d^{n+1}+d} d^{(n+1)d+n d^n} (d+1)^{(n-1)(n+2)} \right)}
{\ln n - \ln (n-1)}.
\end{equation}

(\uppercase\expandafter{\romannumeral 2})
Sufficiency is also obvious.
Now we suppose the minimum of $f$ in $\Delta_n$ is negative,
and a minimal point is $(a_1,\ldots,a_n)$.
From Lemma \ref{lemma:lowerboundf}, we have
$$
|f(a_1,\ldots,a_n)| \geq C_1(M,n,d).
$$

Suppose $(y_{11},\ldots,y_{1n})^T, (y_{21},\ldots,y_{2n})^T \in \Delta_n$,
$(x_{11},\ldots,x_{1n})^T, (x_{21},\ldots,x_{2n})^T$ are coordinates satisfying \eqref{eqn:subs},
from Lemma \ref{normalmatrix}, we have $(x_{11},\ldots,x_{1n})^T, (x_{21},\ldots,x_{2n})^T \in \Delta_n$.
From the correspondence of WDS and barycentric subdivisions, we have
\begin{equation}
\sqrt{(x_{11}-x_{21})^2+\ldots+(x_{1n}-x_{2n})^2} \leq \left(\frac{n-1}{n}\right)^m.\nonumber
\end{equation}
Let $\delta_j = x_{1j} - x_{2j}, j = 1,\ldots,n$, then there exists ~$\gamma \in (0,1)$, such that
\begin{equation}
\begin{split}
& | f(x_{11},\ldots,x_{1n}) - f(x_{21},\ldots,x_{2n}) | \\
= & \left|
\left( \delta_1 \frac{\partial }{\partial x_1} + \ldots
+ \delta_n \frac{\partial }{\partial x_n} \right)
f(x_{21}+ \gamma \delta_1, \ldots, x_{2n} + \gamma \delta_n)
\right| \\
\leq &
\sum\limits_{j=1}^n |\delta_j|
\left|
\frac{\partial }{\partial x_j}
f(x_{21}+ \gamma \delta_1, \ldots, x_{2n} + \gamma \delta_n)
\right|\\
\leq &
\left(\frac{n-1}{n}\right)^m
\sum\limits_{j=1}^n
\sum\limits_{i_1+\ldots+i_n=d}
\left|
c_{i_1 \cdots i_n} (x_{21}+ \gamma \delta_1)^{i_1} \cdots
(x_{2j}+ \gamma \delta_j)^{i_j-1} \cdots
(x_{2n} + \gamma \delta_n)^{i_n}
\right| \\
\leq &
\left(\frac{n-1}{n}\right)^m
\sum\limits_{j=1}^n
\sum\limits_{i_1+\ldots+i_n=d}
M 2^{d-1} \\
= &
n M 2^{d-1} {d+n-1 \choose n-1} \left( \frac{n-1}{n} \right)^m \\
\leq &
n M 2^{d-1} (d+1)^{n-1} \left( \frac{n-1}{n} \right)^m.
\end{split}
\end{equation}
Thus if $m$ is sufficiently large, such that
$$
n M 2^{d-1} (d+1)^{n-1} \left( \frac{n-1}{n} \right)^m \leq
\frac{1}{2} C_1(M,n,d),
$$
i.e.,
\begin{equation} \label{eqn:substimes}
m \geq \dfrac{ \ln \left( 2^{d^n+d} M^{d^n+1} n^{d^{n+1}+d+1} d^{n d^n} (d+1)^{n-1} \right) }
{\ln n - \ln (n-1)},
\end{equation}
$m$-th WDS \eqref{eqn:subs} satisfies
\begin{equation}
f(\alpha_1, \ldots, \alpha_n) \geq \frac{1}{2} C_1(M,n,d).
\end{equation}

From the deduction of (\uppercase\expandafter{\romannumeral 1}), we can see that if
\begin{equation} \label{eqn:c3}
m > \dfrac{ \ln \left( 2^{d^n+1} M^{d^n+1} n^{d^{n+1}+d} d^{(n+1)d+n d^n} (d+1)^{(n-1)(n+2)} \right)}
{\ln n - \ln (n-1)},
\end{equation}
and $m$ satisfies \eqref{eqn:substimes},
there exists a form in $\wds^{(m)}(f)$,
it has complete monomials, and all coefficients are negative.
Comparing the right hand sides of \eqref{eqn:substimes} and \eqref{eqn:c3},
the latter is larger, so it suffices~\eqref{eqn:c3}.

\end{proof}

\end{document}